\documentclass[fleqn,preprint]{elsarticle}

\usepackage{amsmath,amssymb}

\usepackage{tikz}
\usetikzlibrary{calc,positioning,arrows}
\usepackage{graphicx}
\usepackage{float}

\newtheorem{theorem}{Theorem}
\newtheorem{lemma}{Lemma}

\newtheorem{corollary}{Corollary}
\newtheorem{observation}{Observation}

\newtheorem{definition}{Definition}
\newenvironment{proof}{\noindent {\bf Proof.\,}}{\qed}

\def\NP{\mathsf{NP}}
\def\Po{\mathsf{P}}

\begin{document}

\begin{frontmatter}

\title{Note\\[1ex]
Map graphs having witnesses of large girth%girth constrained witness
}

\author{Hoang-Oanh Le}
\ead{LeHoangOanh@web.de}
\address{Berlin, Germany}

\author{Van Bang Le}
\ead{van-bang.le@uni-rostock.de}
\address{Universit\"at Rostock, Institut f\"ur Informatik, Rostock, Germany}

\begin{abstract}
A half-square of a bipartite graph $B=(X,Y,E_B)$ has one color class of $B$ as vertex set, say $X$; two vertices are adjacent whenever they have a common neighbor in $Y$. 
If $G=(V,E_G)$ is the half-square of a {\em planar\/} bipartite graph $B=(V,W,E_B)$, then $G$ is called a map graph, and $B$ is a witness of $G$.
Map graphs generalize planar graphs, and have been introduced and investigated by Chen, Grigni and Papadimitriou [STOC 1998, J. ACM 2002]. They proved that recognizing map graphs is in $\NP$ by proving the existence of a witness. Soon later, Thorup [FOCS 1998] claimed that recognizing map graphs is in $\Po$, by providing an  $\Omega(n^{120})$-time algorithm for $n$-vertex input graphs.  

In this note, we give good characterizations and efficient recognition for half-squares of bipartite graphs with girth at least a given integer $g\ge 8$. It turns out that map graphs having witnesses of girth at least $g$ are precisely the graphs whose vertex-clique incidence bipartite graph is planar and of girth at least $g$. Our structural characterization implies an $O(n^2m)$-time algorithm for recognizing if a given $n$-vertex $m$-edge graph $G$ is such a map graph.

%\keywords{Map graph \and Half-square \and graph algorithm}
% \subclass{68R10 \and 05C85 \and 68Q25} %Zbl
%\subjclass{F.2.2 [Analysis of Algorithms and Problem Complexity]:
%Nonnumerical Algorithms and Problems; G.2.2 [Discrete Mathematics]: Graph Theory}  
\end{abstract}

\begin{keyword}
Map graph; half-square; graph class; graph algorithm
\end{keyword}

\end{frontmatter}

%%%%%%%%%%%%%%%%%%%%%%%
\section{Introduction and preliminaries}
%%%%%%%%%%%%%%%%%%%%%%%

A {\em map graph\/}, introduced and investigated by Chen et al.~\cite{ChenGP98,ChenGP02}, is the intersection graph of simply-connected and interior-disjoint regions of the Euclidean plane. More precisely, a {\em map\/} of a graph $G=(V,E_G)$ is a function $\cal M$ taking each vertex $v\in V$ to a closed disc homeomorph ${\cal M}(v)$ (the regions) in the plane, such that all ${\cal M}(v)$, $v\in V$, are interior-disjoint, and two distinct vertices $v$ and $v'$ of $G$ are adjacent if and only if the boundaries of ${\cal M}(v)$ and ${\cal M}(v')$ intersect. A map graph is one having a map. 
Map graphs are interesting as they generalize planar graphs in a very natural way. Some applications of map graphs have been addressed in~\cite{ChenHK99}. Papers dealing with hard problems in map graphs include \cite{Chen,DemaineFHT,DemaineH,EickmeyerK17,FominLS}. 
Certain map graphs are related to 1-planar graphs~\cite{Brandenburg15,Brandenburg18,ChenGP06}, a relevant topic in graph drawing.   

In~\cite{ChenGP98,ChenGP02}, the notion of {\em half-squares\/} of bipartite graphs also has been introduced in order to give a graph-theoretical characterization of map graphs. The square of a graph $H$, denoted $H^2$, is obtained from $H$ by adding new edges between two distinct vertices whenever their distance in $H$ is two. Then, $H$ is called a square root of $G=H^2$. 
%Given a graph $G$, it is NP-complete to decide if $G$ is the square of some graph $H$ (\cite{MotwaniS}), even for a split graph $H$ (\cite{LauC}). 
Given a bipartite graph $B=(X,Y,E_B)$, the subgraphs of the square $B^2$ induced by the color classes $X$ and $Y$, $B^2[X]$ and $B^2[Y]$, are called the two {\em half-squares\/} of $B$. It turns out that map graphs are exactly half-squares of {\em planar\/} bipartite graphs~\cite{ChenGP98,ChenGP02}. If $G=(V, E_G)$ is a map graph and $B=(V,W,E_B)$ is a planar bipartite graph such that $G=B^2[V]$, then $B$ is a {\em witness\/} of $G$. It is shown in~\cite{ChenGP98,ChenGP02} that an $n$-vertex graph $G=(V,E_G)$ is a map graph if and only if it has a witness $B=(V,W,E_B)$ with $|W|\le 3n-6$, implying that recognizing map graphs is in $\NP$. Soon later, Thorup~\cite{Thorup98} claimed that recognizing map graphs is in $\Po$, by providing a polynomial-time algorithm. (Thorup did not give the running time explicitly, but it is estimated to be roughly $\Omega(n^{120})$ with $n$ being the vertex number of the input graph; cf.~\cite{ChenGP02}.) Thorup's algorithm is very complex and highly non-combinatorial. Given the very high polynomial degree in Thorup's running time, the most discussed problem concerning map graphs is whether there is a faster recognition algorithm for map graphs. One direction in attacking this problem is to consider map graphs with restricted witness. Recently, in \cite{MnichRS16}, it is shown that map graphs with outerplanar witness and map graphs with tree witness can be recognized in linear time. (We remark that Thorup's algorithm cannot be used to recognize map graphs having witnesses with additional properties.)

In this note we consider map graphs with girth constrained witness. Where, the {\em girth\/} of a graph is the length of a shortest cycle in the graph; the girth of a tree is $\infty$. Our motivation is the observation that every planar graph is a map graph with a witness of girth at least six. Indeed, if $G=(V,E_G)$ is a planar graph, then the subdivision of $G$, i.e., the vertex-edge incidence bipartite graph $B$ obtained from $G$ by replacing each edge by a path of length two, has girth at least six, and clearly $G=B^2[V]$. Note, however, that a clique of arbitrary size is a map graph with a star witness. So, in this context, the vertex-clique incidence bipartite graph is more useful than the subdivision.  

\begin{definition}%[vertex-clique incidence bipartite graph]
\label{def:incidence}
Let $G=(V,E_G)$ be an arbitrary graph.
\begin{itemize}
\item The bipartite graph $S_G=(V,E_G, E_B)$ with $E_B=\{ve\mid v\in V, e\in E_G, v\in e\}$ is the \emph{subdivision}, also the \emph{vertex-edge incidence bipartite graph\/}, of~$G$.
\item Let ${\cal C}(G)$ denote the set of all maximal cliques of $G$. The bipartite graph $B_G=(V,{\cal C}(G),E_B)$ with $E_B=\{vQ\mid v\in V, Q\in {\cal C}(G), v\in Q\}$ is the \emph{vertex-clique incidence bipartite graph} of $G$.
\end{itemize}
\end{definition}

Note that subdivisions and vertex-clique incidence graphs of triangle-free graphs coincide. It is quite easy to see that, for every graph $G$, $G=S_G^2[V]$ and $G=B_G^2[V]$. Thus, in the context of map graphs, it is natural to ask for a given graph $G$ whether $S_G$, respectively, $B_G$, is planar. While it is well known that $S_G$ is planar if and only if $G$ is planar, the situation for $B_G$ is not clear yet. 

\begin{quote}
%\emph{Open Question.} 
Which (map) graphs $G$ have  {\em planar\/} vertex-clique incidence bipartite graphs $B_G$?
\end{quote}

The answer for the case when the vertex-clique incidence bipartite graph is a tree has been recently found by the following theorem.

\begin{theorem}[\cite{LeL17,MnichRS16}]\label{thm:tree}
A connected graph is a map graph with a tree witness if and only if it is a block graph, if and only if its vertex-clique incidence bipartite graph is a tree.
\end{theorem}

Where a {\em block graph\/} is one in which every maximal $2$-connected subgraph (the blocks) is a clique. Equivalently, a block graph is a diamond-free chordal graph. (All terms used will given below.) As a consequence, map graphs with tree witness can be recognized in linear time.

In section~\ref{sec:half-squares} we will characterize half-squares of (not necessarily planar) bipartite graphs of large girth. Our structural results imply efficient polynomial time recognition for these half-squares. In section~\ref{sec:witness} we will consider map graphs having witnesses of large girth. It turns out that map graphs having witnesses of large girth (at least eight) admit a similar characterization as in case of tree witnesses stated in Theorem~\ref{thm:tree} above. As a consequence, we will see that such map graphs can be recognized in cubic time. 

All graphs considered are simple and connected. The complete graph on $n$ vertices, the complete bipartite graph with $s$ vertices in one color class and $t$ vertices in the other color class, the cycle with $n$ vertices are denoted $K_n, K_{s,t}$, and $C_n$, respectively. A $K_3$ is also called a {\em triangle\/}, 
a complete bipartite graph $K_{1,n}$ is also called a \emph{star}. The {\em diamond\/}, denoted $K_4-e$, is the graph obtained from the $K_4$ by deleting an edge.

Let $F$ be a graph. {\em $F$-free graphs\/} are those having no induced subgraphs isomorphic to $F$. {\em Chordal graphs\/} are precisely the $C_k$-free graphs, $k\ge 4$. It is well known that block graphs are precisely the diamond-free chordal graphs.

The neighborhood of a vertex $v$ in $G$, denoted $N_G(v)$, is the set of all vertices in $G$ adjacent to $v$; if the context is clear, we simply write $N(v)$. 

For a subset $U\subseteq V$,
$G[U]$ is the subgraph of $G$ induced by $U$, and $G-U$ stands for $G[V\setminus U]$. We write $B=(X,Y,E_B)$ for bipartite graphs with a bipartition into stable sets (color classes) $X$ and $Y$. 
If $G=(V,E_G)$ is a map graph and $B=(V,W,E_B)$ a witness of $G$, then we call $V$ the set of {\em vertices\/} and $W$ the set of {\em points\/}.  

%%%%%%%%%%%%%%%%%%%%%%%%%%%%%%%%%%%%%%%%%%%%%%%%%%%%%%%%%%%%%%%%%
\section{Half-squares of bipartite graphs with girth constraints}\label{sec:half-squares}
%%%%%%%%%%%%%%%%%%%%%%%%%%%%%%%%%%%%%%%%%%%%%%%%%%%%%%%%%%%%%%%%%

Recall that every graph is a half-square of a bipartite graph with girth at least six. 
In this section, we will characterize those graphs that are half-squares of bipartite graphs with large girth. The following facts are easy to verify. 

\begin{observation}\label{obs:cycle}
Let $B=(V,W,E_B)$ be a (not necessarily planar) bipartite graph and let $G=B^2[V]$. Then
\begin{itemize}
\item[\em (i)] 
every induced cycle $C_\ell$, $\ell\ge 4$, in $G$ stems from an induced cycle $C_{2\ell}$ in $B$. That is, if $C_\ell = v_1v_2\ldots v_\ell v_1$ is an induced cycle in $G$, then there are points $w_1, w_2 \ldots, w_{\ell}$ in $W$ such that $v_1w_1v_2w_2 \ldots v_\ell w_\ell v_1$ is an induced cycle in $B$;
\item[\em (ii)] 
every triangle $C_3$ in $G$ stems from an induced $C_6$ or from a $K_{1,3}$ in $B$.
\end{itemize}
\end{observation}

In this section, we prove the following characterizations of half-squares of girth constrained bipartite graphs.

\begin{theorem}\label{thm:girth2t}
Let $t\ge 4$ be an integer. The following statements are equivalent for every graph $G=(V,E_G)$.
\begin{itemize}
\item[\em (i)]
$G$ is half-square of a bipartite graph with girth at least $2t$;
\item[\em (ii)] 
$G$ is diamond-free and $C_\ell$-free for every $4\le \ell\le t-1$; %(C_4,\ldots, C_{h -1})$-free. 
\item[\em (iii)]
The vertex-clique incidence bipartite graph $B_G$ of $G$ has girth at least $2t$. 
\end{itemize}
\end{theorem}
In particular, 
\begin{itemize}
\item a graph is half-square of a bipartite graph with girth at least eight if and only if it is diamond-free, and 
\item a graph is half-square of a bipartite graph with girth at least ten if and only if it is diamond-free and $C_4$-free.
\end{itemize}
It is also interesting to observe that, as $t$ grows larger, Theorem~\ref{thm:girth2t} gets closer and closer to Theorem~\ref{thm:tree} on map graphs with tree witness. (Recall that block graphs are diamond-free and $C_\ell$-free for all $\ell \ge 4$, and a connected graph of girth $\infty$ is a tree.)

\smallskip
\begin{proof}[of Theorem~\ref{thm:girth2t}]

\noindent
(i) $\Rightarrow$ (ii): Let $G=B^2[V]$ for some bipartite graph $B=(V,W,E_B)$ of girth at least $2t$. Then, by Observation~\ref{obs:cycle} (i), $G$ cannot contain any induced cycle $C_\ell$ for $4\le \ell\le t-1$ (otherwise, $B$ would contain an induced cycle of length $2\ell < 2t$). Furthermore, by Observation~\ref{obs:cycle} (ii), $G$ cannot contain an induced diamond. Otherwise, $B$ would contain an induced cycle of length $6<2t$ (if one of the two triangles of the diamond stems from a $C_6$ in $B$) or an induced cycle of length $4<2t$ (if both triangles of the diamond stem from stars in $B$).

\smallskip
\noindent
(ii) $\Rightarrow$ (iii): Let $G$ have no induced diamond and no induced $C_\ell$, $4\le\ell\le t-1$. Recall that $G=B_G^2[V]$. For notational simplicity, write $B=B_G$.  

\smallskip
\noindent
{\em Claim 1}.\, $B$ is $C_4$-free. If not, let $w_1, w_2\in {\cal C}(G)$ be two points belonging to an induced $C_4$ in $B$. Since $N_B(w_1)$ and $N_B(w_2)$ are two distinct maximal cliques in $G$, there is a vertex $u\in N_B(w_1)\setminus N_B(w_2)$ not adjacent to a vertex $v\in N_B(w_2)\setminus N_B(w_1)$ in $G$. But then $u, v$, and two vertices in $N_B(w_1)\cap N_{B}(w_2)$ induce a diamond in $G$. 

\smallskip
\noindent
{\em Claim 2}.\, $B$ is $C_6$-free. If not, let $v_1w_1v_2w_2v_3w_3v_1$ be an induced $C_6$ in $B$. Then, in $G$, $v_3$ is adjacent to two vertices $v_1, v_2$ of the maximal clique $Q=N_B(w_1)$. Since $v_3\not\in Q$, there exists a vertex $u$ in $Q\setminus\{v_1,v_2\}$ not adjacent to $v_3$ in $G$. But then $G$ contains an induced diamond induced by $v_1,v_2,v_3$ and $u$. 

\smallskip
Now, consider a shortest cycle in $B$, $C=v_1w_1v_2w_2 \ldots v_\ell w_\ell v_1$, of length $2\ell$. By Claims~1 and~2, $\ell\ge 4$. Moreover, by the minimality of $C$, $N_{B}(v_i)\cap N_{B}(v_j)=\emptyset$ for $|i-j|>1$ (indices taken modulo $\ell$). Thus, $v_1,v_2\ldots, v_\ell$ induce an cycle of length $\ell$ in $G$, hence $\ell\ge t$. That is, $B$ has girth $2\ell\ge 2t$.

\smallskip
\noindent
(iii) $\Rightarrow$ (i): This implication is obvious as for any graph $G$, $G=B_G^2 [V]$. 
%\qed
\end{proof}

\medskip
Theorem~\ref{thm:girth2t} has algorithmic implications for recognizing half-squares of girth constrained bipartite graphs. 
In the remainder of this note, $n$ and $m$ denotes the vertex number and the edge number, respectively, of the graphs considered.

\paragraph{Half-squares of bipartite graphs of girth at least eight} 
These graphs are precisely the diamond-free graphs, and hence can be recognized in $O(m^{1.5})$ time~\cite{EisenbrandG04}. Note that in a diamond-free graph, each edge belongs to exactly one maximal clique, hence there are at most $m$ maximal cliques.  Since all maximal cliques in a graph can be listed in time $O(nm)$ per generated clique~\cite{TsukiyamaIAS77,ChibaN85}, we can list all maximal cliques in $G$ in time $O(m^2n)$. Thus, assuming $G$ is diamond-free,  the vertex-clique incidence bipartite graph $B_G$ of $G$ can be constructed in time $O(m^2n)$. Thus, we obtain
\begin{corollary}\label{cor:girth8}
Given an $n$-vertex $m$-edge graph $G$, it can be recognized in $O(m^{1.5})$ time if $G$ is half-square of a bipartite graph $B$ with girth at least $8$. If so, such a bipartite graph $B$ can be constructed in $O(m^2n)$ time. 
\end{corollary}

\paragraph{Half-squares of bipartite graphs of girth at least $2t$}  
Let $t\ge 5$ be an integer. Half-squares of bipartite graphs of girth at least $2t$ are precisely the graphs $G$ with $B_G$ is of girth at least $2t$. Note that such graphs $G$ are diamond-free and $C_4$-free (Theorem~\ref{thm:girth2t} (ii)). Thus, we first recognize if the given graph $G$ is diamond-free and $C_4$-free in time $O(m^{\frac{2}{3}}n)$~\cite{EschenHSS11}. If so, we list all maximal cliques in $G$, there are at most $O(n\sqrt{n})$, in the same time complexity~\cite{EschenHSS11}, to construct the vertex-clique incidence bipartite graph $B_G$ of $G$. Since $B_G$ has $n+O(n\sqrt{n})=O(n\sqrt{n})$ vertices, the girth of (the bipartite graph) $B_G$ can be computed in $O((n\sqrt{n})^2)=O(n^3)$ time~\cite{ItaiR78,YusterZ97}. Thus we obtain
  
\begin{corollary}\label{cor:girth2t}
Given an $n$-vertex graph $G$ and an integer $t\ge 5$, it can be recognized in $O(n^3)$ time if $G$ is half-square of a bipartite graph $B$ with girth at least $2t$. If so, such a bipartite graph $B$ can be constructed in the same  time complexity.  
\end{corollary}

It is interesting to note that half-squares of bipartite graphs with girth at least ten, i.e., (diamond, $C_4$)-free graphs, have been very recently discussed in the context of social networks: In~\cite{FoxRSWW18}, {\em $c$-closed graphs\/} are introduced as those graphs, in which every two vertices with at least $c$ common neighbors are adjacent. Thus, $2$-closed graphs are precisely the (diamond, $C_4$)-free graphs, i.e., the half-squares of bipartite graphs with girth at least ten.

%%%%%%%%%%%%%%%%%%%%%%%%%%%%%%%%%%%%%%%%%%%%%%%%%%%%%%%%%%%%%%%%%
\section{Map graphs having witnesses with girth constraints}\label{sec:witness}
%%%%%%%%%%%%%%%%%%%%%%%%%%%%%%%%%%%%%%%%%%%%%%%%%%%%%%%%%%%%%%%%%

While {\em any\/} graph is the half-square of a bipartite graph of girth at least six, not every {\em map graph\/} has a witness of girth at least six. Such a map graph admitting no witness of girth at least six is depicted in Figure~\ref{fig:girth4} below.

\begin{figure}[H]
\begin{center}
\begin{tikzpicture}[scale=.5]
%\tikzstyle{vertex}=[draw,circle,minimum size=10pt,inner sep=.5pt]; 
\tikzstyle{vertex}=[draw,circle, text width=2.7mm, align=center, inner sep=0pt]; 
\tikzstyle{point}=[draw,rectangle,inner sep=2.5pt]; 

\node[vertex] (1) at (0,5)  {\tiny 1}; 
\node[vertex] (2) at (6,5)  {\tiny 2};  
\node[vertex] (3) at (6,1)  {\tiny 3};  
\node[vertex] (4) at (0,1)  {\tiny 4}; 
\node[vertex] (a) at (3,4)  {\tiny a};  
\node[vertex] (b) at (3,2)  {\tiny b};  
\node[vertex] (c) at (8,3)  {\tiny c};  

\draw (1) -- (2) -- (3) -- (4) -- (1);
\draw (a) -- (b) -- (c) -- (a);
\draw (1) -- (a) -- (2); \draw (3) -- (a) -- (4);
\draw (1) -- (b) -- (2); \draw (3) -- (b) -- (4);
\draw (1) to [out=45, in=80] (c);
\draw (c) -- (2); 
\draw (3) -- (c);
\draw (c) to [out=-80, in=-45] (4);
\end{tikzpicture}
\qquad
\begin{tikzpicture}[scale=.5]
%\tikzstyle{vertex}=[draw,circle,minimum size=11pt,inner sep=1pt]; 
\tikzstyle{vertex}=[draw,circle, text width=2.7mm, align=center, inner sep=0pt]; 
\tikzstyle{point}=[draw,rectangle,inner sep=2.5pt]; 

\node[vertex] (1) at (0,5)  {\tiny 1}; 
\node[vertex] (2) at (6,5)  {\tiny 2}; 
\node[vertex] (3) at (6,1)  {\tiny 3}; 
\node[vertex] (4) at (0,1)  {\tiny 4}; 
\node[vertex] (a) at (3,4)  {\tiny a}; 
\node[vertex] (b) at (3,2)  {\tiny b};  
\node[vertex] (c) at (8,3)  {\tiny c};  

\node[point] (w1) at (3,6)  {}; 
\node[point] (w2) at (5,3)  {}; 
\node[point] (w3) at (3,0)  {}; 
\node[point] (w4) at (1,3)  {}; 

\draw (1) -- (w1) -- (2) -- (w2) -- (3) -- (w3) -- (4) -- (w4) -- (1);
\draw (w2) -- (a) -- (w4) -- (b) -- (w2) -- (c);
\draw (w1) to [out=40, in=70] (c);
\draw (w3) to [out=-40, in=-70] (c);
\end{tikzpicture}
\end{center}
\caption{A map graph (left) and a witness of it (right). It can be verified by inspection that this map graph has no witness of girth at least six.\label{fig:girth4}}
\end{figure}
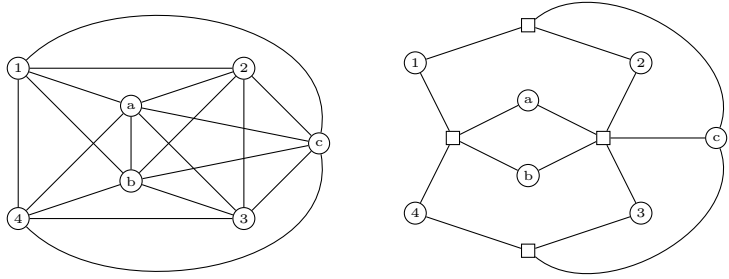

Notice that {\em every planar\/} graph is a map graph with a witness of girth at least six (e.g., its subdivision). So, map graphs admitting witnesses of girth at least six properly include all planar graphs. We are not able to characterize and recognize map graphs having witnesses of girth at least six. 

In this section, we deal with map graphs having witnesses of girth at least eight. Note that not every planar graph has a witness of girth at least eight; the graph depicted in Figure~\ref{fig:diamondfree} is diamond-free, hence it is the half-square of a bipartite graph of girth at least eight. However, as we will see (Theorem~\ref{thm:witness-girth2t}), this diamond-free {\em planar\/} graph does {\em not\/} have a {\em witness\/} of girth at least eight.  

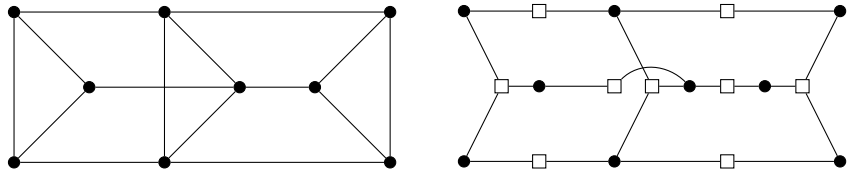
\begin{figure}[htb]
\begin{center}
\begin{tikzpicture}[scale=.5]
\tikzstyle{vertex}=[draw,circle,inner sep=1.5pt, fill=black]; 
\tikzstyle{point}=[draw,rectangle,inner sep=2.5pt]; 

\node[vertex] (1) at (0,0)   {}; 
\node[vertex] (2) at (4,0)   {}; 
\node[vertex] (3) at (10,0)  {}; 
\node[vertex] (4) at (2,2)   {}; 
\node[vertex] (5) at (6,2)   {}; 
\node[vertex] (6) at (8,2)   {};  
\node[vertex] (7) at (0,4)   {}; 
\node[vertex] (8) at (4,4)   {}; 
\node[vertex] (9) at (10,4)  {}; 

\draw (1) -- (2) -- (3) -- (6) -- (9) -- (8) -- (7) -- (1) -- (4) -- (7);
\draw (4) -- (5) -- (6);
\draw (2) -- (5) -- (8) -- (2);
\draw (3) -- (9);
\end{tikzpicture}
\qquad
\begin{tikzpicture}[scale=.5]
\tikzstyle{vertex}=[draw,circle,inner sep=1.5pt, fill=black]; 
\tikzstyle{point}=[draw,rectangle,inner sep=2.5pt]; 

\node[vertex] (1) at (0,0)   {}; 
\node[vertex] (2) at (4,0)   {}; 
\node[vertex] (3) at (10,0)   {}; 
\node[vertex] (4) at (2,2)   {}; 
\node[vertex] (5) at (6,2)   {}; 
\node[vertex] (6) at (8,2)  {};  
\node[vertex] (7) at (0,4)   {}; 
\node[vertex] (8) at (4,4)   {}; 
\node[vertex] (9) at (10,4)   {}; 

\node[point] (w1) at (1,2)   {}; 
\node[point] (w2) at (5,2)   {}; 
\node[point] (w3) at (9,2)   {}; 
\node[point] (w4) at (2,0)   {}; 
\node[point] (w5) at (7,0)   {}; 
\node[point] (w6) at (4,2)  {};  
\node[point] (w7) at (7,2)  {}; 
\node[point] (w8) at (2,4)   {}; 
\node[point] (w9) at (7,4)   {}; 

\draw (1) -- (w1) -- (4); \draw (7) -- (w1);
\draw (2) -- (w2) -- (5); \draw (8) -- (w2);
\draw (3) -- (w3) -- (6); \draw (9) -- (w3);

\draw (1) -- (w4) -- (2) -- (w5) --(3);
\draw (4) -- (w6); \draw (w6) to [out=45, in=135] (5); 
\draw (5) -- (w7) -- (6);
\draw (7) -- (w8) -- (8) -- (w9) --(9);
\end{tikzpicture}

\end{center}
\caption{A diamond-free planar graph (left) with non-planar vertex-clique incidence bipartite graph (right).\label{fig:diamondfree}}
\end{figure}

In discussing witnesses of girth at least eight, the following fact is useful; 
it has been observed and proved in~\cite{LeL17}. 
To make the paper self-contained, we include the proof here.
\begin{lemma}\label{lem:star-clique}
Let $G=B^2[V]$ for some (not necessarily planar) bipartite graph $B=(V,W,E_B)$. If $B$ has no induced cycle of length six, then every maximal clique $Q$ in $G$ stems from a star in $B$, i.e., there is a point $w\in W$ such that $Q=N_B(w)$.
\end{lemma}
\begin{proof}
Suppose to the contrary that there is some clique $Q$ in $B^2[V]$ such that, for any point $w\in W$, $Q\setminus N_B(w)\not=\emptyset$. Choose a point $w_1\in W$ where $Q':=Q\cap N_B(w_1)$ is maximal. Let $v_1\in Q\setminus N_B(w_1)$. Since $Q$ is a clique in $B^2[V]$, there is a point $w_2\in W$ adjacent to $v_1$ and some vertices in $Q'$. Choose such a point $w_2$ with $Q'\cap N_B(w_2)$ is maximal. 
By the choice of $w_1$, there is a vertex $v_2\in Q'\setminus N_B(w_2)$. Again, since $Q$ is a clique in $B^2[V]$, there is a point $w_3\in W$ adjacent to both $v_1$ and $v_2$. By the choice of $w_2$, there is a vertex $v_3\in Q'\cap N_B(w_2)$ non-adjacent to $w_3$. 
But then $w_1, v_2, w_3, v_1, w_2$ and $v_3$ induce a $C_6$ in $B$, a contradiction. Thus, there must be a point $w\in W$ such that $Q\subseteq N_B(w)$, and therefore by the maximality of $Q$, $Q=N_B(w)$.
%\qed
\end{proof}

\medskip
In this section, we prove the following Theorem~\ref{thm:witness-girth2t}, characterizing map graphs having witnesses of large girth. Basically, Theorem~\ref{thm:witness-girth2t} is Theorem~\ref{thm:girth2t} with the additional planarity condition on the vertex-clique incidence bipartite graph. 
\begin{theorem}\label{thm:witness-girth2t}
Let $t\ge 4$ be an integer. The following statements are equivalent for every graph $G=(V,E_G)$.
\begin{itemize}
\item[\em (i)]
$G$ is a map graph having a witness of girth at least $2t$.
\item[\em (ii)] 
$G$ is diamond-free and $C_\ell$-free for every $4\le \ell\le t-1$, and the vertex-clique incidence bipartite graph $B_G$ of $G$ is planar. 
\item[\em (iii)]
The vertex-clique incidence bipartite graph $B_G$ of $G$ is planar and has girth at least $2t$.
\end{itemize}
\end{theorem}
\begin{proof}%[of Theorem~\ref{thm:witness-girth2t}]

\noindent
(i) $\Rightarrow$ (ii): Let $B=(V,W,E_B)$ be a planar bipartite graph of girth at least $2t$ and with minimal number $|W|$ of points such that $G=B^2[V]$. 
First, it follows from Theorem~\ref{thm:girth2t}, part (i) $\Rightarrow$ (ii), that $G$ is diamond-free and $C_\ell$-free for every $4\le \ell\le t-1$. 
Next, by Lemma~\ref{lem:star-clique}, for every maximal clique $Q$ in $G$ there is a point $w_Q\in W$ of $B$ such that $Q=N_B(w_Q)$. Since $|W|$ is minimal, such a point is unique: If there were another point $w'\in W\setminus\{w_Q\}$ with $Q=N_B(w')$, then $B':= B-w'$ would be a planar bipartite graph of girth at least $2t$ and still satisfy $G=B'^2[V]$ with fewer number of points than $B$. Thus, 
\[\phi:B_G\to B, \text{ with $\phi(v)=v$, $v\in V$, and $\phi(Q)=w_Q$, $Q\in{\cal C}(G)$},\] 
is an injective function. Moreover, 
\[vQ\in E_{B_G} \Leftrightarrow v\in Q=N_B(w_Q) \Leftrightarrow \phi(v)\phi(Q)=vw_Q\in E_{B_G},\]
that is, $B_G$ is isomorphic to an induced subgraph of $B$. Hence $B_G$ is planar. 
(In fact, as $G=B_G^2[V]$, $B_G$ is indeed isomorphic to $B$.)     

\smallskip
\noindent
(ii) $\Rightarrow$ (iii): This implication follows immediately from Theorem~\ref{thm:girth2t}, part (ii) $\Rightarrow$ (iii).

\smallskip
\noindent
(iii) $\Rightarrow$ (i): This implication is obvious as $G=B_G^2[V]$.
%\qed
\end{proof}

\medskip
Note that Theorem~\ref{thm:witness-girth2t} is not true in case $t=3$: The graph on the left side of Figure~\ref{fig:diamondfree}, as a planar graph, does admit a witness of girth six but its vertex-clique incidence bipartite graph is not planar; it contains $K_{3,3}$ as a minor.
 
Like Theorem~\ref{thm:girth2t}, one may observe the \lq convergence behavior\rq\ of Theorem~\ref{thm:witness-girth2t} depending on $t$; it \lq converges\rq\ to Theorem~\ref{thm:tree} as $t$ goes to infinity. We now derive efficient recognition algorithm from Theorem~\ref{thm:witness-girth2t} for map graphs having witnesses of large girth. We use the fact that map graphs with $n$ vertices have at most $27\cdot n$ many maximal cliques~\cite{ChenGP02}.

\paragraph{Map graphs having witnesses of girth at least eight} 
Let $G$ be a map graph with a witness of girth at least $8$. 
Since all maximal cliques in a graph can be listed in time $O(nm)$ per generated clique~\cite{TsukiyamaIAS77,ChibaN85}, we can list all maximal cliques in $G$ in time $27n\cdot O(nm)$ to obtain the vertex-clique incidence  bipartite graph $B_G$. Since $B_G$ has $O(n)$ vertices, checking planarity of $B_G$ can be done in $O(|V_{B_G}|+|E_{B_G}|)=O(n^2)$ time, and computing the girth of (planar) $B_G$ can be done in $O(n\log\log n)$ time~\cite{LackiS11}. 
Thus, we obtain
\begin{corollary}\label{cor:witness-girth8}
Map graphs having witnesses of girth at least $8$ can be recognized in time $O(n^2m)$. If so, such a witness can be constructed with the same time complexity.
\end{corollary}

\paragraph{Map graphs having witnesses of girth at least $2t$} 
Let $t\ge 5$ be an integer. Map graphs having witnesses of girth at least $2t$ are particularly diamond-free and $C_4$-free. Thus, we first recognize if the given graph $G$ is diamond-free and $C_4$-free in time $O(m^{\frac{2}{3}}n)$~\cite{EschenHSS11}. If so, we list all maximal cliques in $G$, there are at most $27\cdot n$, in the same time complexity~\cite{EschenHSS11}, to construct the vertex-clique incidence bipartite graph $B_G$ of $G$. Since $B_G$ has at most $28\cdot n$ vertices, checking planarity of $B_G$ can be done in $O(|V_{B_G}|+|E_{B_G}|)=O(n^2)$ time, and computing the girth of (planar) $B_G$ can be done in $O(n\log\log n)$ time~\cite{LackiS11}. Thus we obtain
\begin{corollary}\label{cor:witness-girth2t}
Let $t\ge 5$ be an integer. Map graphs having witnesses of girth at least $2t$ can be recognized in time $O(\max\{m^{\frac{2}{3}}n, n^2\})$. If so, such a witness can be constructed with the same time complexity.
\end{corollary}

%%%%%%%%%%%%%%%%%%%%%
\section{Conclusion} 
%%%%%%%%%%%%%%%%%%%%%

In this note we consider half-squares of girth constrained bipartite graphs. Given an integer $g\ge 8$, we characterize and efficiently recognize half-squares of bipartite graphs with girth at least $g$. We show that map graphs having witnesses with girth at least $g$ are exactly the graphs for which the vertex-clique incidence bipartite graph is planar and of girth at least $g$. Hence map graphs having witnesses of girth at least $g$ can be recognized in $O(n^2 m)$ time. The recognition and characterization problems of map graphs having witnesses of girth at least six remain unsolved. 

\begin{itemize}
\item Which map graphs admit witnesses of girth at least $6$? Besides planar graphs, as pointed out by a referee, the so-called triangulated NIC-planar graphs, investigated in~\cite{BachmaierBHNR17}, are examples of map graphs admitting witnesses of girth six.

\item Is there an efficient combinatorial polynomial-time recognition algorithm for map graphs having witnesses of girth at least $6$?
\end{itemize}

Recall that the class of map graphs having witnesses of girth at least $6$ strictly lies between planar graphs and map graphs, and thus it is an interesting graph class.


\begin{thebibliography}{99}
%==========================

\bibitem{BachmaierBHNR17}
{Christian Bachmaier, Franz J. Brandenburg, Kathrin Hanauer, Daniel Neuwirth, Josef Reislhuber}. 
{NIC-planar graphs}. 
{\em Discrete Applied Mathematics\/} 232 (2017) 23-40, doi: {10.1016/j.dam.2017.08.015}. 

\bibitem{Brandenburg15} Franz J. Brandenburg. 
On 4-map graphs and 1-planar graphs and their recognition problem. 
ArXiv, 2015. http://arxiv.org/abs/1509.03447.

\bibitem{Brandenburg18}
{Franz J. Brandenburg}. 
{Recognizing optimal 1-planar graphs in linear time}. 
{\em Algorithmica\/} 80 (2018) 1-28, doi: {10.1007/s00453-016-0226-8}. 


\bibitem{Chen} Zhi-Zhong Chen. 
Approximation algorithms for independent sets in map graphs. 
{\em J. Algorithms\/} 41 (2001) 20-40, doi: 10.1006/jagm.2001.1178.

\bibitem{ChenGP98} Zhi-Zhong Chen, Michelangelo Grigni, Christos H. Papadimitriou. 
Planar map graphs. 
In {\em Proceedings of the Thirtieth Annual {ACM} Symposium on the Theory of Computing\/} ({STOC} 1998) 514-523, doi: 10.1145/276698.276865.

\bibitem{ChenGP02} Zhi-Zhong Chen, Michelangelo Grigni, Christos H. Papadimitriou. 
Map graphs. 
{\em J. ACM} 49 (2002) 127-138, doi: 10.1145/506147.506148.

\bibitem{ChenGP06}
Zhi{-}Zhong Chen, Michelangelo Grigni, Christos H. Papadimitriou.  
Recognizing hole-free 4-map graphs in cubic time. 
{\em Algorithmica} 45 (2006) 227-262, doi: {10.1007/s00453-005-1184-8}. 

\bibitem{ChenHK99}
Zhi-Zhong Chen,  Xin He, Ming-Yang Kao. 
Nonplanar topological inference and political-map graphs. 
In {\em Proceedings of the Tenth Annual ACM-SIAM Symposium on Discrete Algorithms\/} (SODA 1999) 195-204. 

\bibitem{ChibaN85} Norishige Chiba, Takao Nishizeki. 
Arboricity and subgraph listing algorithms. {\em SIAM J. Comput.\/} 14 (1985) 210-223, doi: 10.1137/0214017.

\bibitem{DemaineFHT} Erik D. Demaine, Fedor V. Fomin, MohammadTaghi Hajiaghayi, Dimitrios M. Thilikos. 
Fixed-parameter algorithms for $(k,r)$-center in planar graphs and map graphs. 
{\em ACM Trans. Algorithms\/} 1 (2005) 33-47, doi: 10.1145/1077464.1077468.

\bibitem{DemaineH} Erik D. Demaine, MohammadTaghi Hajiaghayi. 
The bidimensionality theory and its algorithmic applications. 
{\em Computer J.\/} 51 (2008) 292-302, doi: 10.1093/comjnl/bxm033.

\bibitem{EickmeyerK17}
Kord Eickmeyer, Ken{-}ichi Kawarabayashi. 
{FO} model checking on map graphs. 
In: {\em Fundamentals of Computation Theory - 21st International Symposium\/} 
({FCT} 2017) 204-216, doi: 10.1007/978-3-662-55751-8\textunderscore 17.

\bibitem{EisenbrandG04} 
Friedrich Eisenbrand, Fabrizio Grandoni. 
On the complexity of fixed parameter clique and dominating set. 
{\em Theor. Comput. Sci\/} 326 (2004) 57-67, doi: 10.1016/j.tcs.2004.05.009.

\bibitem{EschenHSS11} 
Elaine M. Eschen, Ch{\'i}nh T. Ho{\`a}ng, Jeremy P. Spinrad, R. Sritharan. 
On graphs without a $C_4$ or a diamond. 
{\em Discrete Applied Mathematics\/} 159 (2011) 581-587, doi: 10.1016/j.dam.2010.04.015.

\bibitem{FominLS} 
Fedor V. Fomin, Daniel Lokshtanov, Saket Saurabh. 
Bidimensionality and geometric graphs. 
In: {\em Proceedings of the Twenty-Third Annual ACM-SIAM Symposium on Discrete Algorithms\/} 
(SODA 2012) 1563-1575, doi: 10.1137/1.9781611973099.124.

\bibitem{FoxRSWW18} 
Jacob Fox, Tim Roughgarden, C. Seshadhri, Fan Wei, Nicole Wein. 
Finding cliques in social networks: A new distribution-free model. 
In: {\em Proceedings of the 45th International Colloquium on Automata, Languages, and Programming\/}
({ICALP} 2018) 1-15, doi: 10.4230/LIPIcs.ICALP.2018.55.

%ArXiv, 2018. https://arxiv.org/abs/1804.07431.

\bibitem{ItaiR78}
Alon Itai, Michael Rodeh. 
Finding a minimum circuit in a graph. 
{\em SIAM J. Comput.\/} 7 (1978) 413-423, doi: 10.1137/0207033.

\bibitem{LackiS11}
%Jakub {\L}{\k{a}}cki, Piotr Sankowski. 
Jakub {\L}acki, Piotr Sankowski.
Min-cuts and shortest cycles in planar graphs in $O(n \log\log n)$ time. 
In: {\em Proceedings of the 19th Annual European Symposium\/} (ESA 2011), 
{\em Lecture Notes in Computer Science\/} 6942 (2011) 155-166, doi: 10.1007/978-3-642-23719-5\textunderscore 14.

\bibitem{LeL17} Hoang-Oanh Le, Van Bang Le. 
Hardness and structural results for half-squares of restricted tree convex bipartite graphs. 
{\em Algorithmica}, in press, doi: 10.1007/s00453-018-0440-7.

\bibitem{MnichRS16} Matthias Mnich, Ignaz Rutter, Jens M. Schmidt. 
Linear-time recognition of map graphs with outerplanar witness. 
{\em Discrete Optimization\/} 28 (2018) 63-77, doi: 10.1016/j.disopt.2017.12.002.


\bibitem{Thorup98} Mikkel Thorup. 
Map graphs in polynomial time. 
In: {\em Proceedings of the 39th IEEE Symposium on Foundations of Computer Science\/} 
(FOCS 1998) 396-405, doi: 10.1109/SFCS.1998.743490.

\bibitem{TsukiyamaIAS77}
{Shuji Tsukiyama, Mikio Ide, Hiromu Ariyoshi and Isao Shirakawa}. 
{A new algorithm for generating all the maximal independent sets}. 
{\em SIAM J. Comput.\/} 6 (1977) 505-517, doi: 10.1137/0206036.

\bibitem{YusterZ97}
Raphael Yuster, Uri Zwick. 
Finding even cycles even faster. 
{\em SIAM J. Discrete Math.\/} 10 (1997) 209-222, doi: 10.1137/S0895480194274133.
%Girth in O(nm). bipartite girth in O(n^2)
\end{thebibliography}
\end{document}